\documentclass[a4paper,UKenglish,cleveref,autoref,thm-restate]{lipics-v2021}

\pdfoutput=1 
\hideLIPIcs  

\usepackage{latexsym,amsmath,textcomp,pifont,marvosym,wasysym,amssymb}
\usepackage{graphicx}
\usepackage{xcolor}
\usepackage{mathtools}
\usepackage{booktabs}

\usepackage{numprint}

\usepackage{tikz}
\usetikzlibrary{calc}

\usepackage{pgfplots}
\usepgfplotslibrary{external}
\tikzsetexternalprefix{plots/}

\newif\ifpdfplots
\pdfplotstrue

\usepackage{algorithmic}
\usepackage[algo2e, ruled, vlined, linesnumbered]{algorithm2e}
\DontPrintSemicolon
\SetKwProg{Fn}{Function}{}{end}
\SetKwFor{ParallelWhile}{while}{do in parallel}{}
\SetKwFor{ParallelDo}{parallel do}{}{}
\SetKwFor{ParallelFor}{for}{do parallel}{}
\SetKwInput{KwInput}{Input}
\SetKwInput{KwOutput}{Output}

\newcommand{\pluseqatomic}{\overset{\text{\tiny atomic}}{\mathrel{+}=}}

\title{Linear-Time Multilevel Graph Partitioning via Edge Sparsification}

\author{Lars Gottesbüren}{Google Research, Zürich, Switzerland}{gottesbueren@google.com}{}{}

\author{Nikolai Maas}{Karlsruhe Institute of Technology, Karlsruhe, Germany}{nikolai.maas@kit.edu}{}{}

\author{Dominik Rosch}{Karlsruhe Institute of Technology, Karlsruhe, Germany}{dominik.rosch@student.kit.edu}{}{}

\author{Peter Sanders}{Karlsruhe Institute of Technology, Karlsruhe, Germany}{sanders@kit.edu}{}{}

\author{Daniel Seemaier}{Karlsruhe Institute of Technology, Karlsruhe, Germany}{daniel.seemaier@kit.edu}{}{}

\authorrunning{L. Gottesbüren and N. Maas and D. Rosch and P. Sanders and D. Seemaier}

\Copyright{Lars Gottesbüren and Nikolai Maas and Dominik Rosch and Peter Sanders and Daniel Seemaier} 
\category{} 
\relatedversion{} 
\supplement{\url{ae.iti.kit.edu/documents/research/linear_time_mgp/}}
\funding{\flag[2.5cm]{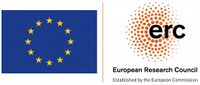}This project has received funding from the European Research Council (ERC) under the European Union’s Horizon 2020 research and innovation program (grant agreement No. 882500).}
\nolinenumbers 

\EventEditors{John Q. Open and Joan R. Access}
\EventNoEds{2}
\EventLongTitle{42nd Conference on Very Important Topics (CVIT 2016)}
\EventShortTitle{CVIT 2016}
\EventAcronym{CVIT}
\EventYear{2016}
\EventDate{December 24--27, 2016}
\EventLocation{Little Whinging, United Kingdom}
\EventLogo{}
\SeriesVolume{42}
\ArticleNo{23}

\newcommand{\Oh}[1]{\ensuremath{\mathcal{O}(#1)}}
\newcommand{\Th}[1]{\ensuremath{\Theta(#1)}}

\newcommand{\pluseq}{\mathrel{+}=}

\definecolor{darkgreen}{rgb}{0.0, 0.6, 0.0}

\newcommand{\splitatcommas}[1]{%
  \begingroup
  \begingroup\lccode`~=`, \lowercase{\endgroup
    \edef~{\mathchar\the\mathcode`, \penalty0 \noexpand\hspace{0pt plus 1em}}%
  }\mathcode`,="8000 #1%
  \endgroup
}

\newcommand{\Partitioner}[1]{\textsc{#1}}
\newcommand{\Instance}[1]{\textsf{#1}}

\newcommand{\DensityBound}{\tau_d}
\newcommand{\EdgeBound}{\tau_e}
\newcommand{\ReductionFactor}{\rho}

\begin{document}

\maketitle

\begin{abstract}
The current landscape of balanced graph partitioning is divided into high-quality but expensive multilevel algorithms and cheaper approaches with linear running time, such as single-level algorithms and streaming algorithms.
We demonstrate how to achieve the best of both worlds with a \emph{linear time multilevel algorithm}.
Multilevel algorithms construct a hierarchy of increasingly smaller graphs by repeatedly contracting clusters of nodes.
Our approach preserves their distinct advantage, allowing refinement of the partition over multiple levels with increasing detail.
At the same time, we use \emph{edge sparsification} to guarantee geometric size reduction between the levels and thus linear running time.

We provide a proof of the linear running time as well as additional insights into the behavior of multilevel algorithms, showing that graphs with low modularity are most likely to trigger worst-case running time.
We evaluate multiple approaches for edge sparsification and integrate our algorithm into the state-of-the-art multilevel partitioner \Partitioner{KaMinPar}, maintaining its excellent parallel scalability.
As demonstrated in detailed experiments, this results in a $1.49\times$ average speedup (up to $4\times$ for some instances) with only 1\% loss in solution quality.
Moreover, our algorithm clearly outperforms state-of-the-art single-level and streaming approaches.

\end{abstract}

\newpage

\section{Introduction}

Balanced graph partitioning aims to divide a graph into blocks of roughly equal size while minimizing the number of edges cut by the partition.
As this is a crucial subtask in many applications~\cite{GRAPH-SURVEY, HYPERGRAPH-SURVEY}, it is of considerable interest to compute high-quality partitions within a minimal amount of time.
Unfortunately, this goal seems unattainable from the viewpoint of complexity theory -- even approximating balanced graph partitioning to a constant factor is NP-hard~\cite{AndreevRaeckeApprox}.
Consequently, heuristic approaches are used in practice, covering a wide spectrum of options along the running time versus quality trade-off.

In the high-quality category, the most successful approaches use the multilevel framework.
By repeatedly contracting clusters of nodes, multilevel algorithms first construct a hierarchy of increasingly smaller graphs in the \emph{coarsening phase}.
On the smallest graph, more expensive heuristics can be used to find a good \emph{initial partition}.
Finally, the \emph{uncoarsening phase} undoes the contractions in reverse order while further improving the partition quality via local search algorithms (this is called \emph{refinement}).
Overall, multilevel partitioning combines a good initial solution with iterative refinement on a series of summarized graph representations with increasingly finer granularity.
This has proven highly successful in practice, consistently achieving better solution quality than alternative approaches on real-world inputs~\cite{GRAPH-RECENT-SURVEY, mt-kahypar-journal, DBLP:journals/siamsc/KarypisK98}.
However, due to lacking constraints on the size of the contracted representations, current multilevel implementations have superlinear running time.

On the other hand, \emph{single-level} algorithms that use only a fixed number of passes on the input graph can run in linear time~\cite{Pulp}.
This is motivated by applications where the partitioning time is a potential bottleneck.
For example, graph partitioning is used in various domains to efficiently distribute workloads across parallel machines~\cite{MeshPartitioningApplication, MatrixComputationsApplication, GraphProcessingApplication}.
This requires the graph partitioning step to be less expensive than the downstream computation.
Taking this to the extreme, \emph{streaming} approaches only consider a small part of the graph at once, assigning nodes greedily while using only a minimal representation of the partition state~\cite{StreamCPI, HeiStream, Cuttana}.
However, the running time guarantees of single-level and streaming algorithms come at the cost of inferior solution quality when compared to multilevel algorithms~\cite{Restreaming, mt-kahypar-journal, Fennel}.

\subparagraph{Contributions.}
In this work, we show that the described trade-off can be avoided by constructing a \emph{linear time multilevel algorithm}.
Our coarsening algorithm enforces that the graph shrinks by a constant factor with every successive contraction step, using edge sparsification to reduce the number of edges if necessary.
We prove that this guarantees $\Oh{n + m}$ expected total work for $n$ nodes and $m$ edges, without any assumptions on the input graph.
Our analysis provides a framework to understand the running time behavior of a broad class of existing multilevel algorithms.
In addition, we demonstrate that graphs with low modularity are most likely to trigger worst-case running time behavior, while graphs with high modularity might already allow linear running time without using edge sparsification.

We integrate our approach into the \Partitioner{KaMinPar} shared-memory graph partitioner~\cite{deep-mgp}, preserving its excellent scaling behavior while guaranteeing linear work.
For instance classes that approximate the worst case, our algorithm achieves practical speedups of up to 4$\times$ (1.49$\times$ in the geometric mean) over a baseline \Partitioner{KaMinPar} configuration -- which is the fastest available shared-memory multilevel partitioner according to Ref.~\cite{mt-kahypar-journal}.
Despite this, the loss in partition quality is only 1\% on average.
Our algorithm outperforms both the single-level partitioner \Partitioner{PuLP}~\cite{Pulp} and the state-of-the-art streaming partitioner \Partitioner{CUTTANA}~\cite{Cuttana}, achieving 24\% and 66\% smaller average cuts, respectively, as well as a faster running time.

\section{Preliminaries}\label{sec:preliminaries}

\subparagraph*{Notation and Definitions.}
Let $G = (V, E, c, \omega)$ be an undirected graph with node weights $c: V \rightarrow \mathbb{N}_{> 0}$, edge weights $\omega: E \rightarrow \mathbb{N}_{> 0}$, $n \coloneqq \lvert V \rvert$, and $m \coloneqq \lvert E \rvert$.
We extend $c$ and $\omega$ to sets, i.e., $c(V' \subseteq V) \coloneqq \sum_{v \in V'} c(v)$ and $\omega(E' \subseteq E) \coloneqq \sum_{e \in E'} \omega(e)$.
$N(v) \coloneqq \{ u \mid \{u, v\} \in E\}$ denotes the neighbors of $v \in V$ and $E(v) \coloneqq \{ e \mid v \in e \}$ denotes the edges incident to $v$.
We are looking for $k$ \emph{blocks} of nodes $\Pi \coloneqq \{ V_1, \dots, V_k \}$ that partition $V$, i.e., $V_1 \cup \dots \cup V_k = V$ and $V_i \cap V_j = \emptyset$ for $i \neq j$.
The \emph{balance constraint} demands that $\forall i \in \{1, \dots, k\}$: $c(V_i) \le L_{\max} \coloneqq (1 + \varepsilon) \lceil \frac{c(V)}{k} \rceil$ for some imbalance parameter $\varepsilon > 0$.
The objective is to minimize $\mathrm{cut}(\Pi) \coloneqq \sum_{i < j} \omega(E_{ij})$ (weight of all cut edges), where $E_{ij} \coloneqq \{ \{u, v\} \in E \mid u \in V_i, v \in V_j \}$.
A \emph{clustering} $\mathcal{C} \coloneqq \{ C_1, \dots, C_b \}$ is also a partition of $V$, where the number of blocks $b$ is not given in advance (there is also no balance constraint).

\subparagraph*{Multilevel Graph Partitioning.}
Virtually all high-quality, general-purpose graph partitioners are based on the multilevel paradigm, which consists of three phases.
During coarsening, the algorithms construct a hierarchy $\mathcal{H} = \langle G \eqqcolon G_1, G_2, \dots, G_\ell \rangle$ of successively coarser representations of the input graph $G$.
Coarse graphs are built by either computing node clusterings or matchings and afterwards \emph{contracting} them.
A clustering $\mathcal{C} = \{C_1, \dots, C_b\}$ is contracted by replacing each cluster $C_i$ with a coarse node $c_i$ with weight $c(c_i) = c(C_i)$.
For each pair of clusters $C_i$ and $C_j$, there is a coarse edge $e = \{c_i, c_j\}$ with weight $\omega(e) = \omega(E_{ij})$ if $E_{ij} \neq \emptyset$, where $E_{ij}$ is the set of all edges between clusters $C_i$ and $C_j$.
Once the number of coarse nodes falls below a threshold (typically, $kC$ for some tuning constant $C$), \emph{initial partitioning} computes an initial solution of the coarsest graph $G_\ell$.
Subsequently, contractions are undone, projecting the current solution to finer graphs and refining it.
The total running time of a multilevel partitioner is the cumulative time for coarsening, initial partitioning, and refinement across all levels of the hierarchy $\mathcal{H}$.

\section{Related Work}\label{sec:related}

There has been a lot of research on graph partitioning, thus we refer the reader to surveys~\cite{GRAPH-SURVEY, GRAPH-RECENT-SURVEY} for a general overview and only focus on work closely related to our contributions here.
As described above, modern general-purpose, high-quality graph partitioners such as \Partitioner{Mt-Metis}~\cite{MT-METIS}, \Partitioner{Mt-KaHIP}~\cite{MT-KAHIP}, \Partitioner{Mt-KaHyPar}~\cite{mt-kahypar-journal}, \Partitioner{KaMinPar}~\cite{deep-mgp}, and \Partitioner{Jet}~\cite{Jet-Journal} are mostly based on the multilevel paradigm, which constructs a hierarchy of coarser graphs during the coarsening phase.

\subparagraph*{Graph Coarsening.}
Early multilevel partitioners, like \Partitioner{Chaco}~\cite{MISC:conf/supercomp/HendricksonL95} and \Partitioner{Metis}~\cite{DBLP:journals/siamsc/KarypisK98}, primarily employed coarsening strategies based on contracting graph matchings.
While effective for mesh-like graphs due to high matching coverage (often 85-95\%~\cite{AnalysisMGP}), 
these strategies struggle with graphs exhibiting irregular structures, such as scale-free networks.
On these graphs, small maximal matchings can result in much slower coarsening and potentially a linear number of levels.
Subsequent developments addressed this limitation.
\Partitioner{Mt-Metis}~\cite{DBLP:conf/sc/LaSallePSSDK15} introduced $2$-hop matchings, extending small maximal matchings by further pairing nodes that have some degree of overlap in their neighborhoods until $\ge 75\%$ of nodes are contracted.
This technique was subsequently also implemented by other partitioners~\cite{Mongoose, Jet-Journal}.
Alternative strategies focus on accelerating coarsening by grouping multiple nodes.
These include methods based on cluster contraction~\cite{DBLP:conf/wea/MeyerhenkeSS14, MT-KAHIP, mt-kahypar-journal, deep-mgp} and pseudo-matchings where nodes can match with multiple neighbors~\cite{MGPPowerLawGraphs}.
While enabling faster node reductions, they often constrain the weight of the clusters to ensure that finding a balanced initial partition is feasible.
This can be problematic on graphs with highly connected hubs (e.g., the center of a star graph), potentially limiting the achievable coarsening ratio.

\subparagraph*{Graph Sparsification.}
Graph sparsification techniques aim to approximate a given graph with a sparser one (called \emph{sparsifier}), typically containing substantially fewer edges while preserving specific structural properties important for downstream tasks.
This allows handling massive data sets where considering the full graph is computationally infeasible, as well as speeding up a variety of algorithms on graphs or matrices~\cite{SamplingSurvey, benczur96, 10.1007/s00453-022-01053-4}.
For graph partitioning, preserving cut properties (and thus approximately preserving the partition objective) is particularly relevant.
An \emph{$\varepsilon$-cut sparsifier} guarantees that every cut in the sparsifier has a weight within a $1 \pm \varepsilon$ factor of the original cut.
Bencz\'ur~and~Karger showed that such sparsifiers with $O(n \log n / \varepsilon^2)$ edges exist for any graph and gave near-linear time constructions~\cite{benczur96}.

There are several approaches to construct sparsifiers.
Spielman~and~Srivastava~\cite{DBLP:journals/siamcomp/SpielmanS11} introduced sparsification based on \emph{effective resistance}, which often yields high-quality sparsifiers and preserves spectral properties closely related to cuts.
However, this method can be computationally demanding.
Alternatively, various heuristic sampling techniques exist, such as uniform edge sampling, $k$-neighbor sampling, and Forest Fire sampling~\cite{ForestFireSampling, NetworKitForestFire}, which uses an analogy to a spreading wildfire.
Chen~et~al.~\cite{DBLP:journals/pvldb/ChenYVBDMT23} provide a comparative study, suggesting that Forest Fire sampling outperforms uniform sampling for preserving cut-related properties.

\subparagraph*{KaMinPar.}
We integrate the techniques described in this paper into the \Partitioner{KaMinPar}~\cite{deep-mgp} framework.
\Partitioner{KaMinPar} is a shared-memory parallel multilevel graph partitioner.
Its coarsening and uncoarsening phases are based on the \emph{size-constrained label propagation}~\cite{DBLP:conf/wea/MeyerhenkeSS14} algorithm, which is parameterized by a maximum cluster size (resp. block weight) $U$.
In the coarsening resp. uncoarsening phase, each node is initially assigned to its own cluster resp. to its corresponding block of the partition.
The algorithm then proceeds in rounds.
In each round, the nodes are visited in some order.
A node $u$ is moved to the cluster resp. block $K$ that contains the most neighbors of $u$ without violating the size constraint $U$, i.e., $c(K) + c(u) \le U$.
The algorithm terminates once no nodes have been moved during a round or a maximum number of rounds has been exceeded.
The coarsening further implements a 2-hop clustering strategy~\cite{deep-mgp}, which reduces the number of coarse nodes further whenever label propagation alone yields a node reduction factor less than $2$.
Since each round of size-constrained label propagation runs in linear time, and there is only a constant number of rounds, \Partitioner{KaMinPar} achieves linear time per hierarchy level for coarsening and uncoarsening.

The original paper~\cite{deep-mgp} shows that \Partitioner{KaMinPar} achieves overall linear-time complexity under two key assumptions:
(i) a constant node reduction factor between hierarchy levels, and (ii) bounded average degree for coarse graphs.
While we will demonstrate in \Cref{thm:twohop_analysis} that \Partitioner{KaMinPar}'s coarsening strategy satisfies assumption (i), the inability to guarantee assumption (ii) results in a worst-case running time with an extra $log(n)$ factor.

\section{Linear Time Multilevel Graph Partitioning}\label{sec:linear-time}

Multilevel algorithms construct a hierarchy $\mathcal{H} = \langle G \eqqcolon G_1, G_2, \dots, G_\ell \rangle$ of successively coarser representations of the input graph $G$.
Each level of $\mathcal{H}$ is considered twice,
during coarsening (to construct the next level) and during refinement (to improve the current partition).
Assuming linear time for the coarsening and refinement on each level (see Section~\ref{sec:related}), the total sequential running time is $\Th{\sum_{i=1}^\ell |V_i| + |E_i|}$.
Without additional constraints on the number and size of the levels, the worst-case running time might be $\Th{n m}$ or worse.

To obtain better guarantees, we need a geometric size reduction per level.\footnote{In general, any series with a sum of $\Oh{1}$ works -- a geometric series is, however, the most straightforward.}
As a first step, we require that $|V_{i+1}| \le \gamma |V_i|$ for some constant $\gamma < 1$ that is independent of $G$.
If this is the case, the coarsened graph has constant size after a logarithmic number of steps, which already achieves a running time of $\Oh{n + m \log n}$.
Combined with a similar guarantee for the number of edges, we get a linear total running time.

\subsection{Reducing the Number of Nodes}\label{sec:coarsening-nodes}

As discussed in Section~\ref{sec:related}, the coarsening algorithms used in practice start by computing either a matching or a clustering of adjacent nodes.
Typically, a maximum allowed node weight $U$ is enforced for clusters.
We use the term \emph{size-constrained label propagation} to refer to a broad class of coarsening algorithms that form clusters of adjacent nodes and use a weight constraint.
We require one essential property.
In the resulting clustering, a node $v$ never forms a singleton cluster as long as there is any adjacent cluster $K$ with $c(K) + c(v) \le U$.

Due to the weight constraint, size-constrained label propagation by itself is not sufficient for reducing the number of nodes (consider, e.g., a star graph).
To solve this, partitioners use \emph{2-hop clustering} as a second step, forming clusters of nodes that are not adjacent but instead have a common neighbor cluster.
In the following, we provide the first formal proof that this guarantees a constant factor node reduction.

Consider a (non-isolated) node $v$ in a singleton cluster $S = \{v\}$.
Formally, we will assume that the algorithm assigns a \emph{favorite cluster} $K_S$ to $S$, out of the clusters adjacent to $S$.
The 2-hop clustering then merges any nodes with the same favorite cluster, as long as this does not violate the weight constraint.
Note that only considering favorite clusters is more restrictive than general 2-hop clustering, but is already sufficient for our purpose.

\begin{figure}
    \centering
    \includegraphics[width=0.63\linewidth]{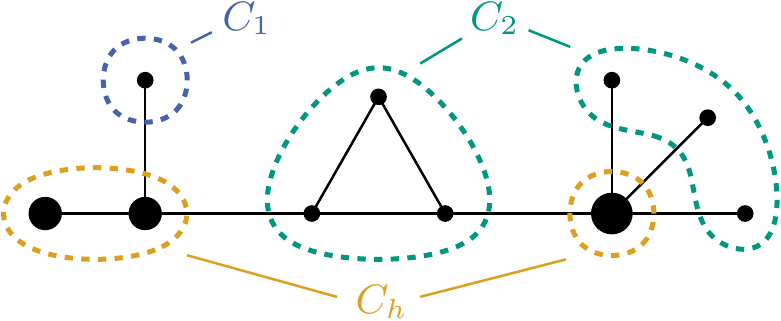}
    \caption{Illustration of Theorem~\ref{thm:twohop_analysis}, with examples of the different cluster types. Note that the green cluster to the right is created by 2-hop clustering.}
    \label{fig:theorem_illustration}
\end{figure}

\begin{theorem}
    \label{thm:twohop_analysis}
    The number of clusters obtained by size-constrained label propagation and 2-hop clustering with a maximal cluster weight $U \geq 2 \frac{c(V)}{|V|}$ is at most
    \[
        \lvert \mathscr{C} \rvert \le \frac{1}{2} |V| + \frac{c(V)}{U}
    \]
    on any graph without isolated nodes.
\end{theorem}

\begin{proof}
    We divide the set of clusters $\mathscr{C}$ into multiple subsets (see Figure~\ref{fig:theorem_illustration} for an illustration).
    $C_h$ is the set of \emph{heavy} clusters with weight larger than $\frac{1}{2} U$.
    $C_1$ is the set of singleton clusters with weight at most $\frac{1}{2} U$ and $C_2$ is the set of clusters with multiple nodes and weight at most $\frac{1}{2} U$.
    Note that $\mathscr{C} = C_h \cup C_1 \cup C_2$.
    Let $r := \frac{1}{\lvert C_2 \rvert} \sum_{K \in C_2} |K|$ be the average number of nodes for clusters in $C_2$.
    In combination, this results in the inequality $\lvert V \rvert \ge \lvert C_h \rvert + \lvert C_1 \rvert + r \lvert C_2 \rvert$.

    Each singleton cluster $S \in C_1$ is only adjacent to clusters with weight larger than $U - c(S)$, and thus only clusters in $C_h$ -- otherwise, the node would have joined the lighter adjacent cluster.
    Consider the favorite cluster $K_S \in C_h$ of $S$. Due to 2-hop coarsening, there is no other cluster in $C_1$ with the same favorite (otherwise, 2-hop clustering would have joined them).
    Consequently, $|C_1| \le |C_h|$.
    Moreover, $c(S) + c(K_S) > U$ gives, when summed over all clusters and combined with the definition of $C_h$, the inequality $c(V) \ge \sum_{K \in C_h} c(K) + \sum_{K \in C_1} c(K) = \sum_{K \in C_1} (c(K) + c(K_S) ) + \sum_{K \in C_h \setminus \{ K'_S \mid K' \in C_1 \}} c(K) > U |C_1| + \frac{1}{2} U (|C_h| - |C_1|)$.
    Rearranged, this is $|C_h| + |C_1| \le 2 \frac{c(V)}{U}$.

    Combining all inequalities, we get
    \begin{align*}
        \lvert \mathscr{C} \rvert &= |C_h| + |C_1| + |C_2|\\
        &\le \frac{1}{r} |V| + (1 - \frac{1}{r}) |C_h| + (1 - \frac{1}{r}) |C_1|\\
        &\le \frac{1}{r} |V| + 2 (1 - \frac{1}{r}) \frac{c(V)}{U}\\
        &\le \frac{1}{2} |V| + \frac{c(V)}{U}
    \end{align*}
    For the final step, we use the observation that $x a + (1-x)b \le \frac{1}{2}a + \frac{1}{2}b$ for $b \le a$ and $x \le \frac{1}{2}$.
    Since $r \ge 2$ and $U \ge 2 \frac{c(V)}{|V|}$, we can apply this with $x = \frac{1}{r}$, $a = |V|$ and $b = 2\frac{c(V)}{U}$.
\end{proof}

Isolated nodes (i.e., nodes without a neighbor) are a special case as standard clustering algorithms do not handle them.
Therefore, they are omitted from Theorem~\ref{thm:twohop_analysis}.
However, it is trivial to either remove isolated nodes and reinsert them in the uncoarsening, or alternatively cluster them with each other (we do the latter).

Note that the precondition $U \geq 2 \frac{c(V)}{|V|}$ is no limitation for the applicability of Theorem~\ref{thm:twohop_analysis}.
In practice, much larger values are used for $U$ (in our case $U = \frac{c(V)}{160k}$, see Section~\ref{sec:complete-algorithm}).
However, the theorem does not include clustering approaches which limit the \emph{number} of nodes in a cluster (e.g., allowing only matchings).
We note that in this case similar, but weaker, bounds can be obtained with an analogous line of reasoning.

\subsection{Reducing the Number of Edges via Sparsification}\label{sec:sparsification}

\begin{algorithm2e}[t]
    \caption{Graph Coarsening with Sparsification.}
    \label{alg:coarsening}

    $i \gets 1$, $G_i \gets G$ \tcp*{Input: graph $G$}
    \While{$G_i$ not small enough}{
        $G'_{i + 1} \gets \FuncSty{Coarsen}(G_i)$ \; \label{line:coarsen}
        $\hat{m} \gets \min\{ \EdgeBound \cdot \lvert E(G_i) \rvert, \DensityBound \cdot \frac{\lvert E(G_i) \rvert}{\lvert V(G_i) \rvert} \cdot \lvert V(G'_{i + 1}) \rvert \}$ \;
        \lIf{$\lvert E(G'_{i + 1}) \rvert > \ReductionFactor \cdot \hat{m}$}{ \label{line:trigger}
            $G_{i + 1} \gets \FuncSty{Sparsify}(G'_{i + 1}, \hat{m})$ \label{line:sparsification}
        }
        \lElse
        {
            $G_{i + 1} \gets G'_{i + 1}$ 
        }
        $i \pluseq 1$ \;
    }
    \Return{$\langle G_1, \dots, G_{i} \rangle$} \tcp*{Output: hierarchy $\mathcal{H} = \langle G_1, \dots, G_i \rangle$}
\end{algorithm2e}

\begin{figure}[t]
    \centering
    \includegraphics{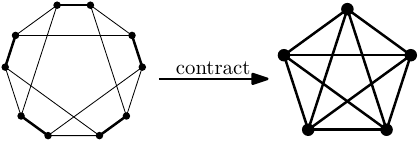}
    \caption{Contracting the bolded edges leads to increased density on the coarse graph.}
    \label{fig:density_increase_example}
\end{figure}

As discussed above, coarsening strategies with geometrical node reduction might still exhibit superlinear total running time due to density increase on coarse graphs (e.g., \Cref{fig:density_increase_example}).
For a more general example, consider the coarsening hierarchy of a sparse Erd\H{o}s-R\'enyi graph $G_0 = G(n_0, c / n_0)$ with $n_0$ nodes and edge probability $p_0 \coloneqq c / n_0$ for some constant $c$.
Assume that coarsening halves the number of nodes at each level by contracting pairs of nodes, and that coarse graphs also behave like Erd\H{o}s-R\'enyi graphs.
In other words, $G_i = G(n_i, p_i)$ with $n_i = n_{i - 1} / 2$ and $p_i \approx 1 - (1 - p_{i - 1})^4$ for $i > 0$ (there is an edge between two coarse nodes if any of the four potential edges between the corresponding nodes in $G_{i - 1}$ existed).
Note that $n_i = n_0 / 2^i$ and $p_i = 1 - (1 - p_0)^{4^i} = 1 - (1 - c / n_0)^{4^i} \approx 1 - e^{-4^i \cdot c / n_0}$.
Let $i = \alpha \log(n_0)$, then $p_i \approx 1 - e^{-c n_0^{2 \alpha - 1}} \xrightarrow{n_0 \rightarrow \infty} 0$ for $\alpha < \frac{1}{2}$, i.e., there are $\Theta(\log n_0)$ sparse levels.
On these,
\[
    \frac{\mathbb{E}[m_{i + 1}]}{\mathbb{E}[m_i]} = \frac{1 - (1 - p_i)^4}{p_i} \frac{n_i (n_i - 2) / 8}{n_i (n_i - 1) / 2}
    \xrightarrow{n_0 \rightarrow \infty} \frac{1 - (1 - p_i)^4}{4 p_i} \approx \frac{1 - (1 - 4 p_i)}{4 p_i} = 1
    \text{,}
\]
since $n_i = n_0 / 2^i \ge n_0 / 2^{\alpha \log n_0} > \sqrt{n_0} \rightarrow \infty$ and $(1 - p_i)^4 \approx 1 - 4 p_i$ for small $p_i$.
Thus, the number of coarse edges remains relatively constant, leading to overall $\Oh{m_0 \log(n_0)}$ time.

To achieve linear time, we therefore limit the number of edges through sparsification as outlined in \Cref{alg:coarsening}.
Let $G'_{i + 1} = (V_{i + 1}, E'_{i + 1})$ denote the current graph before sparsification, obtained by contracting the previous graph $G_i$ (line \ref{line:coarsen}).
We obtain $G_{i + 1} = (V_{i + 1}, E_{i + 1})$ by sparsifying the edges of $G'_{i + 1}$ so that the size of $E_{i + 1}$ is bounded by a threshold $\hat{m}$, defined as
\[
    \hat{m} \coloneqq \min\{ \EdgeBound \cdot \lvert E_i \rvert, \DensityBound \cdot \frac{\lvert E_i \rvert}{\lvert V_i \rvert} \cdot \lvert V_{i + 1} \rvert \}
    \text{.}
\]
Here, $\EdgeBound$ is the \emph{edge threshold} parameter, limiting the coarse edge count relative to the current graph's edge count, and $\DensityBound$ is the \emph{density threshold} parameter, likewise limiting the average degree of the coarse graph.
Since sparsification itself introduces computational overhead, we only apply it if the potential edge reduction is significant.
Specifically, we trigger sparsification only if $\lvert E'_{i + 1} \rvert > \hat{m}$ \emph{and} the target edge count $\hat{m}$ represents a substantial reduction from the current edge count $\lvert E'_{i + 1} \rvert$, quantified by the condition $\lvert E'_{i + 1} \rvert / \hat{m} \ge \ReductionFactor$, where $\ReductionFactor \ge 1$ is a tunable constant (line \ref{line:trigger}).
Once triggered, we use one of the following sampling algorithms to reduce the edge count to $\hat{m}$ (in expectation), before adding the sparsified graph to the hierarchy (line \ref{line:sparsification}).
Since our goal is to achieve overall linear time, we only consider linear time sparsification algorithms.
Further, sparsification must be fast in practice for speedups to be attainable.

\vspace{0.8em}\subparagraph*{Uniform Sampling.}
As a simple baseline, we consider uniform random sampling.
Each edge $e \in E'_{i + 1}$ is selected independently with probability $p \coloneqq \hat{m} / \lvert E'_{i + 1} \rvert$, resulting in expected $\hat{m}$ edges.
Note that this approach is oblivious to edge weights -- although heavier edges have larger influence on the partitioning objective and are thus likely more important.

\subparagraph*{Weighted Threshold Sampling.}
To incorporate edge weights, we consider a weighted threshold sampling strategy.
First, we identify the weight threshold $\omega_t \coloneqq \omega(e_t)$ corresponding to the $\hat{m}$-th heaviest edge $e_t$ in $G'_{i + 1}$.
This can be done in expected time $\Oh{\lvert E'_{i + 1} \rvert}$ using the quickselect algorithm.
Based on $\omega_t$, we partition $E'_{i + 1}$ into three disjoint sets $E_{i + 1}^{'<}$, $E_{i + 1}^{'=}$, and $E_{i + 1}^{'>}$, for coarse edges with weight smaller than, equal to, or larger than $\omega_t$.
Edges in $E_{i + 1}^{'<}$ are discarded, while edges in $E_{i + 1}^{'>}$ are kept.
To reach the target size $\hat{m}$, we further sample edges from $E_{i + 1}^{'=}$ uniformly with probability $p \coloneqq \frac{\hat{m} - \lvert E_{i + 1}^{'>} \rvert}{\lvert E_{i + 1}^{'=} \rvert}$.

\subparagraph*{(Weighted) Forest Fire Sampling.}
We further include a variation of threshold sampling that uses \emph{Forest Fire}~\cite{NetworKitForestFire} scores rather than edge weights, as this performed well as a cut-preserving sparsifier in Ref.~\cite{DBLP:journals/pvldb/ChenYVBDMT23}.
We include a brief description for self-containment.
The algorithm computes edge scores by simulating fires spreading through the graph via multiple traversals starting from random nodes.
When visiting a node $u$, the number of neighbors $X$ to be visited is drawn from a geometric distribution parameterized by a tunable parameter $p$.
Subsequently, $X$ distinct unvisited neighbors of $u$ are sampled, incrementing the \emph{burn scores} of the corresponding edges and spreading the fire to the sampled nodes.
The algorithm stops scheduling fires once the cumulative burn score exceeds $\nu \lvert E \rvert$ for some \emph{burn ratio} $\nu > 0$.

We also experiment with a weighted variation of the algorithm, called \emph{Weighted Forest Fire} (WFF), where the probability of sampling a certain neighbor is scaled with the weight of the corresponding edge.
For more details, we refer to \Cref{app:wff}.

\subsection{Putting it Together}\label{sec:complete-algorithm}

Based on the discussed insights, we propose a linear time multilevel algorithm that builds upon \Partitioner{KaMinPar}~\cite{deep-mgp}.
We leverage the existing clustering and refinement algorithms available in \Partitioner{KaMinPar}, whose running time is linear in the size of the current hierarchy level (see Section~\ref{sec:related}).
We introduce two necessary changes to achieve linear time for the overall algorithm.
Most importantly, we introduce edge sparsification as discussed in Section~\ref{sec:sparsification}, ensuring the number of edges shrinks geometrically.
In addition, we replace the coarsening and initial partitioning used by the default configuration of \Partitioner{KaMinPar} with a more traditional approach.
This is because the default configuration is amenable to scenarios where expensive bipartitioning happens on a relatively large graph, adding a $\Omega(n \log n)$ term to the running time in the worst-case.
We refer to Appendix~\ref{app:deep-multilevel} for details.

Instead, we use size-constrained label propagation with subsequent recursive bipartitioning, following other state-of-the-art multilevel algorithms~\cite{MT-KAHIP, mt-kahypar-d, DBLP:conf/wea/MeyerhenkeSS14}.
Similar to \Partitioner{Mt-KaHyPar}~\cite{mt-kahypar-journal}, the cluster weight limit is $U = \frac{c(V)}{160k}$ and we limit the node reduction per coarsening step to at most $2.5\times$.
As this is combined with 2-hop clustering, Theorem~\ref{thm:twohop_analysis} guarantees a geometric shrink factor until a size of $|V_i| = 320k$ is reached.
The coarsening terminates at $160k$ nodes or if the current shrink factor is too small, thereby resulting in a graph with size $\Oh{k}$.\footnote{Note that sparsification ensures $\Oh{k}$ edges -- although this is not necessary for linear time.}
The recursive bipartitioning then requires total time $\Oh{k \log k}$, which is linear under the extremely weak assumption that $k \log k \in \Oh{n + m}$.

So far, we have argued from a sequential point of view.
In the parallel setting, the consequence is that our algorithm needs only linear work.
With regards to scalability, the sparsification algorithms described in Section~\ref{sec:sparsification} lend themselves to a rather straightforward parallelization.
Combined with the excellent scalability of the coarsening and refinement algorithms of \Partitioner{KaMinPar}~\cite{deep-mgp} and the fact that initial partitioning is insignificant for the total running time, we maintain the scalability of default \Partitioner{KaMinPar} while reducing the required work.

\section{Quantifying Worst-Case Instances}

As discussed, we need edge sparsification to achieve linear time if coarsening does not shrink the number of edges geometrically.
However, it would be useful to understand for which graphs sparsification is required and for which it is not -- both for theoretical insights into the structure of worst-case instances and to allow empirical estimates.
Given a clustering $C$ of a graph $G$, we are thus interested in the number of edges of $G'$, where $G'$ is the graph created by contracting $C$.
If $|E(G')| \approx |E(G)|$ for clusterings computed by the coarsening algorithm, sparsification is required to further reduce the number of edges.

Intuitively, this is the case for graphs with low locality -- edges might lead anywhere and are thus hard to contain in small clusters.
For many random graph models, this is rather easy to decide.
For example, sparse Erd\H os-R\'{e}nyi graphs are highly non-local, thus necessitating sparsification (see Section~\ref{sec:sparsification}).
On the other hand, for random geometric graphs (i.e., unit disk graphs) coarsening algorithms reduce the number of edges very efficiently.
However, to classify real-world instances or more complex graph models, a general criterion is needed.

\subparagraph*{Classification via Modularity.}
We propose that the \emph{modularity} of a graph allows to estimate whether sparsification is necessary.\footnote{
There are also multiple other locality metrics, but these are less useful.
For example, the clustering coefficient is based on the number of triangles.
However, this does not result in any useful bounds.
}
Modularity was introduced by Newman and Girvan to evaluate the quality of a clustering with regards to community structure~\cite{ModularityIntroduction},
and modularity based community detection algorithms are used in many applications~\cite{KAHYPAR-CA, ModularityApplication2, ModularityApplication1}.
Given a clustering $C = \{ V_1, \dots, V_{|C|} \}$, let $e_{ij} := \frac{1}{2|E|} |E(V_i, V_j)|$ denote the fraction of edges that connect cluster $i$ and cluster $j$ (only counted in one direction).
Further, let $a_i := \sum_j e_{ij}$ be the fraction of edges with one endpoint in cluster $i$.
The modularity of the clustering is then defined as $Q_C := \sum_i \left( e_{ii} - a_i^2 \right)$, where $Q_C \in [-\frac{1}{2}, 1]$.
The modularity $Q \in [0, 1]$ of the graph itself is the maximum modularity of all possible clusterings.
As demonstrated in the following, modularity is a good fit for our purpose.

\begin{lemma}
    \label{lemma:modularity}
    For a given clustering, the total fraction of edges that connect nodes within the same cluster is bounded by
    \[
        Q_C \le \sum_i e_{ii} \le Q_C + \alpha_C
    \]
    where $\alpha := \max_i a_i$ is the maximum fraction of edges with endpoints in the same cluster.
\end{lemma}

\begin{proof}
    Since $Q = \sum_i \left( e_{ii} - a_i^2 \right)$, the left side of the inequality follows immediately.
    Further, $\sum_i e_{ii} = Q + \sum_i a_i^2 \le Q + \sum_i (\max_j a_j) a_i = Q + \max_i a_i$.
    Note that $\sum_i a_i = 1$ since each edge is connected to exactly two clusters.
\end{proof}

Lemma~\ref{lemma:modularity} provides a lower bound of $1 - Q_C - \alpha_C$ for the fraction of edges connecting different clusters. 
Unfortunately, this does not correspond directly to the edges of $G'$ -- if multiple edges connect the same pair of clusters (we say that the edges are \emph{parallel}), they are combined into a single edge, thereby further reducing the number of remaining edges.
The actual bound is thus $1 - Q_C - \alpha_C - p_C \le |E(G')|$, where $p_C$ is the number of parallel edges.

Let us consider the case where the clusters are small, such as during the first steps of multilevel coarsening.
Then, $1 - Q$ is an approximate lower bound for the number of edges.
If clusters are small, any cluster is only incident to a small fraction of all edges and edges are unlikely to be parallel, i.e., both $\alpha_C$ and $p_C$ are small.
Moreover, $Q_C$ is almost certainly smaller than $Q$ as achieving maximum modularity often necessitates large clusters~\cite{ModularityResolutionLimit}.
Consequently, we expect that $1 - Q \lesssim 1 - Q_C - \alpha_C - p_C \le |E(G')|$ if clusters are small, making $1 - Q$ an accurate bound.

\begin{figure}
    \begin{minipage}{0.49\textwidth}
    \ifpdfplots
    \includegraphics{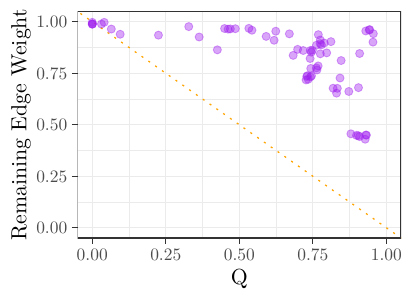}
    \else
    \tikzsetnextfilename{weight_reduction}%
    \input{tikz/weight_reduction}%
    \fi

    \end{minipage}
    \begin{minipage}{0.49\textwidth}
    \ifpdfplots
    \includegraphics{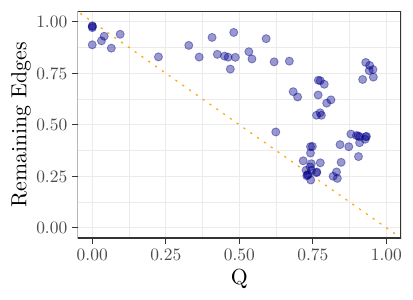}
    \else
    \tikzsetnextfilename{edge_reduction}%
    \input{tikz/edge_reduction}%
    \fi

    \end{minipage}
    \caption{Remaining total edge weight (\textbf{left}) and number of edges (\textbf{right}) after one coarsening step, compared to the modularity of the graph. The $y$-values are denoted as a fraction of the initial value.
    Based on Lemma~\ref{lemma:modularity}, we expect most points (i.e., graphs) to be in the upper right half.
    }
    \label{fig:modularity}
\end{figure}

\subparagraph*{Empirical Effect of Modularity.}
To verify whether $1 - Q$ is a useful bound in practice, we provide an empirical evaluation on a set of 71 large graphs which is used in multiple recent works on graph partitioning~\cite{MtKahyparDeterministicV2, mt-kahypar-ufm, TeraPart}.
Figure~\ref{fig:modularity} shows computed modularity scores of the graphs in relation to the fraction of edges that remains after one step of our multilevel coarsening algorithm (see Section~\ref{sec:coarsening-nodes}).
Since computing the modularity of a graph is itself NP-hard~\cite{ModularityNPHard}, we calculate approximate scores with the well-known Louvain algorithm~\cite{Louvain}, using the implementation from NetworKit~\cite{NetworKitNew, PARALLEL-LOUVAIN}.

With regards to the total edge weight of $G'$ (left, corresponds to inter-cluster edges in $G$), $1 - Q$ is almost a strict lower bound.
For most graphs, the fraction of remaining weight is much larger than $1 - Q$.
The actual number of edges after combining parallel edges (right) is often significantly smaller than the weight.
However, $1 - Q$ is still a mostly accurate bound.
Therefore, modularity is indeed useful to predict the coarsening behavior of multilevel algorithms -- graphs with low modularity are likely to be worst-case instances.

\section{Experiments}

\subparagraph*{Setup.}
We implemented the described sparsification algorithms within the \Partitioner{KaMinPar}~\cite{deep-mgp} framework and compiled it using \texttt{gcc}~14.2.0 with flags \texttt{-O3 -mtune=native -march=native}.
The code is parallelized using TBB~\cite{TBB}.
We compare our algorithm against \Partitioner{PuLP}~\cite{Pulp} (v1.1) and \Partitioner{Cuttana}~\cite{Cuttana} (commit \texttt{ed0c182} in the official GitHub repository\footnote{\url{https://github.com/cuttana/cuttana-partitioner}}), which are linear~time single-level resp. streaming algorithms.
We use the default settings for \Partitioner{PuLP} and configure \Partitioner{Cuttana} following the parameters described by the authors, i.e., $\frac{\mathcal{K'}}{\mathcal{K}} = 4\,096$, $D_{\max} = 1\,000$ and $\textrm{max\_qsize} = 10^6$.
All experiments are performed on a machine equipped with two 32-core Intel Xeon Gold 6530 processors (2.1\,GHz) and 3\,TB RAM running Rocky Linux 9.4.
We only use one of the two processors (i.e., 32 cores) to avoid NUMA effects.

\subparagraph*{Instances.}
We focus on graphs for which coarsening increases edge density substantially, as sparsification is not triggered otherwise (thus leaving the algorithm unchanged).
This happens on 17 out of 71 graphs used for benchmarking in Ref.~\cite{mt-kahypar-ufm} (mostly real-world k-mer and social graphs, and graphs deduced from text recompression~\cite{TextRecompression}).
We further include 6 social graphs from the Sparse~Matrix~Collection~\cite{SPM} and generate random graphs: Erd\H{o}s-R\'enyi graphs (using KaGen~\cite{KaGen}), as well as Chung-Lu~\cite{DBLP:conf/waw/MillerH11}, Planted~Partition, and R-MAT~\cite{DBLP:conf/sdm/ChakrabartiZF04} graphs (using NetworKit~\cite{NetworKitNew}).
These graphs are inherently non-local, thus especially challenging for linear-time partitioning.
Overall, the benchmark set comprises 39 graphs (listed in \Cref{benchmark-table}) with 511\,K to 1.8\,G undirected edges.
The graphs deduced from text recompression feature node weights.
All other graphs are unweighted.
Tuning experiments are performed on a subset containing 8 randomly drawn graphs spanning different types (bolded in \Cref{benchmark-table}).

\subparagraph*{Methodology.}
We consider an \emph{instance} as the combination of a graph and a number of blocks $k$.
We set the imbalance tolerance to $\varepsilon = 3\%$, use $k \in \{3, 7, 8, 16, 37, 64\}$ and perform 5 repetitions for each instance using different seeds.
Results (running time, edge cut) are averaged arithmetically per instance over these repetitions.
When aggregating across multiple instances, we use the geometric mean to ensure that each instance has equal influence. 

\subparagraph*{Performance Profiles.}
To compare the edge cuts of different algorithms, we use \emph{performance profiles}~\cite{PERFORMANCE-PROFILES}.
Let $\mathcal{A}$ be the set of all algorithms we want to compare, $\mathcal{I}$ the set of instances, and $\textrm{cut}_A(I)$ the edge cut of algorithm $A \in \mathcal{A}$ on instance $I \in \mathcal{I}$.
For each algorithm $A$, we plot the fraction of instances 
\[
    \mathcal{P}_A(\tau) \coloneqq \frac{\lvert \{ I \in \mathcal{I} ~:~ \textrm{cut}_A(I) \le \tau \cdot \min_{A' \in \mathcal{A}} \textrm{cut}_{A'}(I) \} \rvert}{\lvert \mathcal{I} \rvert}
\] on the $y$-axis and $\tau$ on the $x$-axis.
Achieving higher fractions at lower $\tau$-values is considered better.
In particular, $\mathcal{P}_A(1)$ denotes the fraction of instances for which algorithm $A$ performs best, while $\mathcal{P}_A(\tau)$ for $\tau > 1$ illustrates the robustness of the algorithm.
For example, an algorithm $A$ with $\mathcal{P}_A(1) = 0.49$ but $\mathcal{P}_A(1.01) = 1.0$ (i.e., never more than 1\% worse than the best) might be preferable to an algorithm $B$ with $\mathcal{P}_B(1) = 0.51$ that only achieves $\mathcal{P}_B(\tau) = 1.0$ at much larger $\tau$ (indicating much worse partitions on some inputs).

\subsection{Parameter Study}

We begin our evaluation by tuning the parameters introduced in \Cref{sec:sparsification}.
Recall that these are the edge and density thresholds $\EdgeBound$ and $\DensityBound$, which control the number of coarse edges, and the minimum reduction factor $\ReductionFactor$, which controls whether sparsification is triggered on a given hierarchy level.
We use the tuning benchmark subset and $k = 16$ for this experiment.

\begin{figure}[t]
    \centering
    \includegraphics{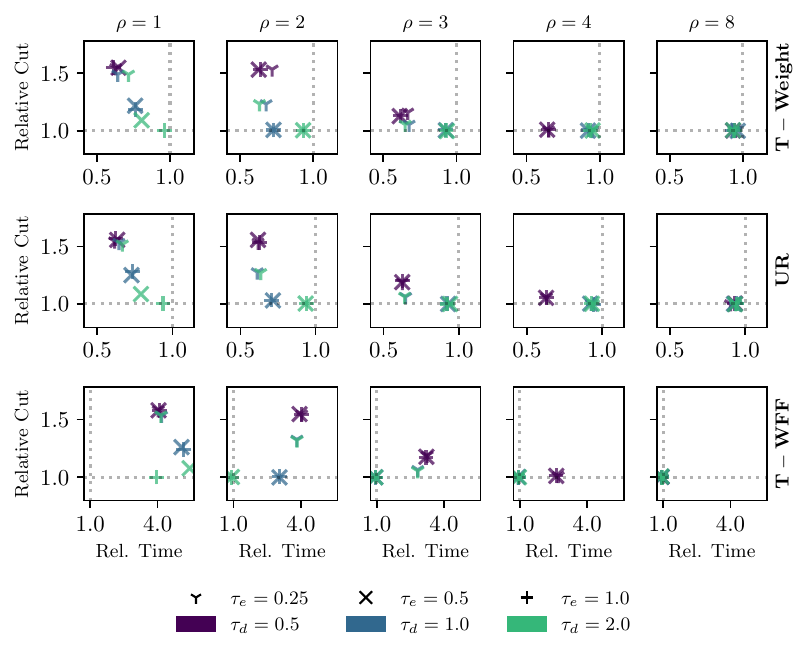}
    \caption{
        Relative cut and running time of \Partitioner{KaMinPar} with weighted threshold sampling (T-Weight), uniform sampling (UR), or threshold sampling via Weighted Forest Fire scores (T-WFF) versus baseline (no sparsification) on the tuning benchmark set with $k = 16$.
    }
    \label{fig:tuning}
\end{figure}

The results are shown in \Cref{fig:tuning}, where we plot geometric mean edge cuts and running times relative to the \Partitioner{KaMinPar} baseline without sparsification for $\EdgeBound \in \{1/4, 1/2, 1\}$, $\DensityBound \in \{1/2, 1, 2\}$ and $\ReductionFactor \in \{1, 2, 3, 4, 8\}$ on the tuning set using $k = 16$ blocks.
We observe similar speedups of up to 1.63$\times$ for weighted threshold sampling (T-Weight) and uniform sampling (UR).
T-Weight achieves the highest speedup (1.63$\times$) at $\EdgeBound = 1/2$, $\DensityBound = 1/2$ and $\ReductionFactor = 3$, while UR achieves the same speedup at slightly different parameters ($\EdgeBound = 1/4$, $\DensityBound = 1$ and $\ReductionFactor = 2$).
With these parameters, edge cuts increase by 12.6\% and 27.6\% for T-Weight and UR, respectively.
Surprisingly, more aggressive sparsification (i.e., smaller $\EdgeBound$, $\DensityBound$ and $\ReductionFactor$) does not achieve larger speedups.
This is likely due to several factors. 
First, the sparsification process itself introduces computational overhead which can counteract potential speedups, particularly when the size reduction is modest.
Second, excessive sparsification degrades partition quality considerably, thereby increasing the workload required for the refinement algorithm to converge to a local optimum.
Hence, moderate sparsification seems favorable.

Larger $\ReductionFactor$ seems beneficial for maintaining partition quality.
At $\ReductionFactor = 4$ (i.e., only sparsify if reducing the number of edges by a factor of $\ge 4$), both T-Weight and UR show similar speedups as with smaller $\ReductionFactor$ and partition quality close to the baseline.
We therefore pick $\EdgeBound = \DensityBound = 1/2$ and $\ReductionFactor = 4$ for subsequent experiments, where T-Weight and UR achieve speedups of 1.56$\times$ and 1.59$\times$, while increasing edge cuts by 0.9\% and 5.3\%, respectively.

Lastly, we look at threshold sampling using Weighted Forest Fire (T-WFF) scores.
For T-WFF itself, we use $p = 0.6$ and $\nu = 0.5$, since these parameters performed best during preliminary experimentation.
We observe the fastest running times using parameters that do not trigger sparsification, suggesting that T-WFF does not provide practical speedups.
At $\EdgeBound = \DensityBound = 1/2$ and $\ReductionFactor = 4$, T-WFF is 2.62$\times$ slower while incurring a 1.0\% increase in cut size.

\subsection{Effects of Sparsification}

\begin{figure}[t]
    \centering
    \begin{minipage}{0.49\textwidth}
        \includegraphics[width=\columnwidth]{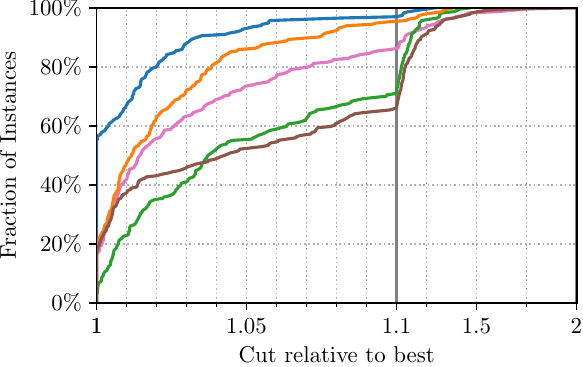}
    \end{minipage}
    \begin{minipage}{0.49\textwidth}
        \includegraphics[width=\columnwidth]{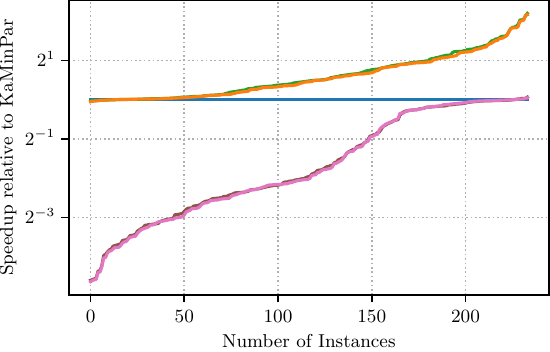}
    \end{minipage}
    \includegraphics[width=\columnwidth]{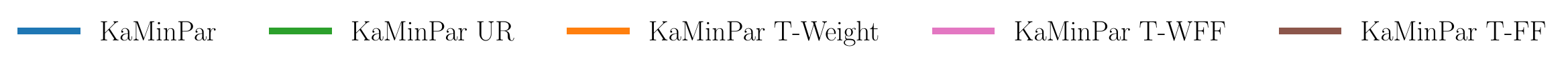}

    \caption{
        Partition quality (\textbf{left}) and speedup over baseline (no sparsification, \textbf{right}) of sparsification algorithms: weighted threshold sampling (T-Weight), uniform sampling (UR), and threshold sampling via (Weighted) Forest Fire scores (T-(W)FF). 
    }
    \label{fig-all-performance-speedup}
\end{figure}

Next, we evaluate the proposed sparsification techniques on the full benchmark set.
As can be seen in \Cref{fig-all-performance-speedup}, T-Weight (geometric mean running time 1.43\,s) and UR (1.40\,s) achieve similar speedups of 1.49$\times$ resp. 1.52$\times$ over the baseline (no sparsification, 2.13\,s).
T-Weight achieves considerably better partition quality (increase in average edge cut by 1.5\%) than UR (increase by 5.5\%).
T-WFF outperforms T-FF, but is not competitive: its partition quality is slightly worse (increase by 3.9\%) while much slower (7.04\,s).
We thus focus on T-Weight.

\begin{figure}[t] 
    \centering
    \includegraphics[width=0.49\textwidth]{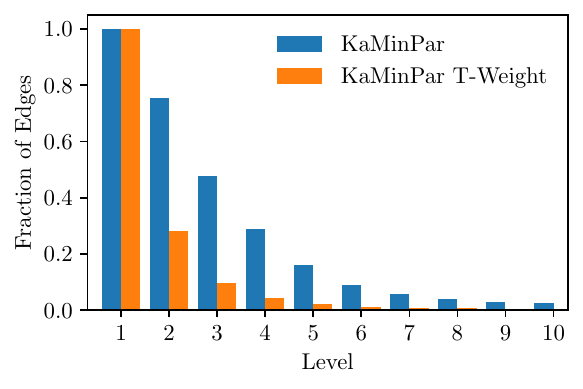}
    \caption{
        Relative geometric mean number of edges per hierarchy level (levels 1-10), comparing no sparsification against T-Weight (weighted threshold sampling) sparsification.
        Edge counts are relative to the input graphs.
        The final value is propagated for hierarchies shorter than 10 levels.
    } 
    \label{fig-m-over-levels}
\end{figure}

Looking at \Cref{fig-m-over-levels}, we can see that sparsification reduces the number of edges on coarse graphs considerably.
Without sparsification, the graphs on the first hierarchy levels (i.e., after the first coarsening step) contain, on average, 75\% of the edges of the input graphs, but only 39\% of the nodes.
With sparsification, the average edge count reduces to 28\%.

\begin{figure}[t]
    \begin{minipage}[c]{0.49\textwidth}
        \includegraphics[width=\columnwidth]{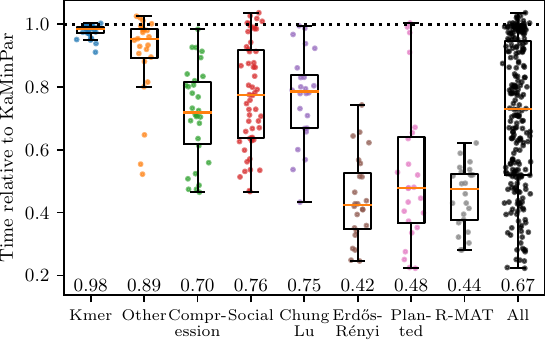}
    \end{minipage}
    \hfill
    \begin{minipage}[c]{0.49\textwidth}
        \includegraphics[width=\columnwidth]{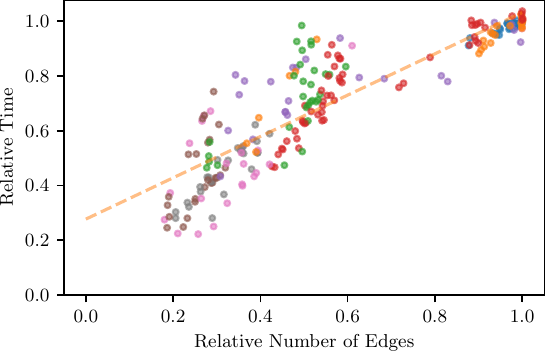}
    \end{minipage}

    \caption{
        Running time of \Partitioner{KaMinPar} with T-Weight sparsification relative to the baseline without sparsification (lower is better).
        \textbf{Left:} Relative running times grouped by graph class (see \Cref{benchmark-table}), with the geometric mean relative time annotated per class.
        \textbf{Right:} Relationship between relative running time and the hierarchy size ratio (number of total edges across all hierarchy levels after sparsification relative to baseline).
        Note the strong correlation (correlation coefficient $\approx 0.893$).
    }
    \label{fig-speedups-per-class}
\end{figure}

As shown in \Cref{fig-speedups-per-class} (left), the speedup from T-Weight sparsification varies considerably with graph class (defined in \Cref{benchmark-table}).
Non-local graphs with extremely low modularity (Erd\H{o}s-R\'enyi, Planted~Partition, R-MAT) show substantial gains (average speedup $> 2\times$, up to $4\times$), while real-world text recompression ($\approx 1.43\times$) and social graphs ($\approx 1.32\times$) exhibit moderate speedups.
K-mer graphs see negligible benefit.
This variation correlates strongly with the reduction in graph hierarchy size (number of edges across all hierarchy levels): instances with greater reduction achieve faster relative running times (\Cref{fig-speedups-per-class}, right).
The observed speedups stem primarily from reduced time in the coarsening and refinement phases, see \Cref{fig-run-time-shares}.
Without sparsification, these phases consume on average 55\% (1.17\,s) and 23\% (0.48\,s) of the total partitioning time (2.13\,s), respectively.
Sparsification reduces these to 0.65\,s and 0.32\,s, respectively.
This improvement comes at low cost, as the sparsification step itself averages only 0.10\,s out of 1.62\,s when triggered (94\% of the instances).

\begin{figure}[t]
    \centering

    \includegraphics{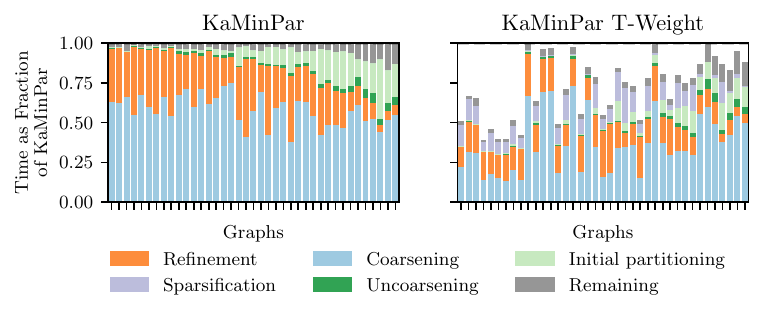}

    \caption{
        Relative running time attribution for \Partitioner{KaMinPar} without sparsification (\textbf{left}) and with T-Weight sparsification (\textbf{right}) using $k = 16$.
        Graphs are sorted by the total running time of \Partitioner{KaMinPar} without sparsification in descending order.
    }
    \label{fig-run-time-shares}
\end{figure}

\subsection{Comparison against Competing Partitioners}

\begin{figure}[t]
    \centering

    \begin{minipage}{0.49\textwidth}
        \includegraphics[width=\columnwidth]{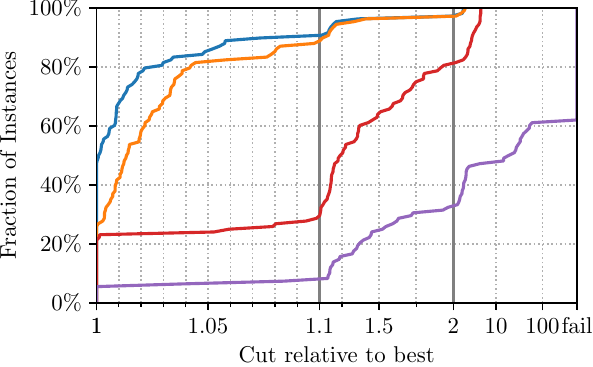}
    \end{minipage}
    \begin{minipage}{0.49\textwidth}
        \includegraphics[width=\columnwidth]{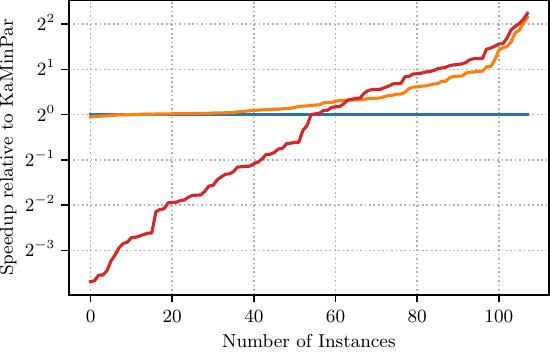}
    \end{minipage}
    \includegraphics[width=\columnwidth]{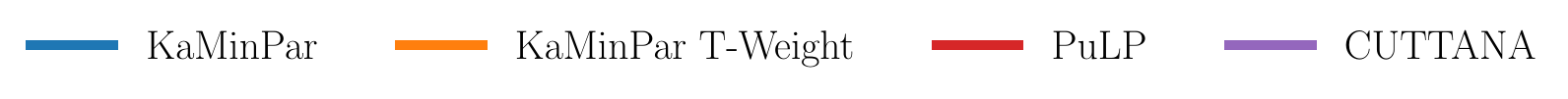}

    \caption{
        Partition quality (\textbf{left}) and relative running time (\textbf{right}) on the reduced benchmark set and all $k$ values for \Partitioner{KaMinPar} without and with T-Weight sparsification, \Partitioner{PuLP} and \Partitioner{Cuttana}.
        Speedups are plotted relative to \Partitioner{KaMinPar} without sparsification.
        \Partitioner{Cuttana} is omitted from the speedup plot since it is $\geq72\times$ slower than all other algorithms.
    }
    \label{fig-competitors}
\end{figure}

Finally, we compare \Partitioner{KaMinPar} with T-Weight sparsification against alternative linear-time partitioners: single-level \Partitioner{PuLP}~\cite{Pulp} and streaming \Partitioner{Cuttana}~\cite{Cuttana}.
We found that \Partitioner{Cuttana} is rather slow and does not support node weights.
Thus, we limit our benchmark set to unweighted graphs with $m \leq 2^{26}$ edges (18 out of 39 graphs).

As shown in \Cref{fig-competitors} (left), \Partitioner{KaMinPar} with T-Weight sparsification computes considerably better partitions with average cuts 30\% and 66\% smaller than those of \Partitioner{PuLP} and \Partitioner{Cuttana}, respectively.
Compared to \Partitioner{KaMinPar} \emph{without} sparsification, edge cuts are slightly larger (cutting 1\% more edges on average), but are still within a factor of 1.10 to the best cut found on 88\% of all instances (vs. 90\% for non-sparsifying \Partitioner{KaMinPar}).
In contrast, \Partitioner{PuLP} and \Partitioner{Cuttana} compute edge cuts within factors 1.26 and 1.93 to the best cuts found on only \emph{half of the instances}, respectively.
\Partitioner{PuLP} computes the best partitions for 21\% of the instances, predominantly Erd\H{o}s-R\'enyi and Planted~Partition graphs.
\Partitioner{Cuttana} crashes on 39\% of the instances (we exclude these instances in pairwise aggregates).

Through sparsification, the geometric mean running time of \Partitioner{KaMinPar} reduces from 0.48\,s to 0.37\,s.
\Partitioner{PuLP} (0.63\,s) is slower than non-sparsifying \Partitioner{KaMinPar} on average, but proves faster on 53 (resp. 48 vs. sparsifying \Partitioner{KaMinPar}) and twice as fast on 23 (resp. 2) out of 108 instances.
\Partitioner{Cuttana} is 72$\times$ slower than \Partitioner{PuLP} (and thus the other algorithms), although we note that comparing running times fairly is difficult since \Partitioner{Cuttana} interleaves computation with graph I/O from SSD (I/O times are excluded for the other algorithms).

\section{Conclusion}

Current graph partitioning algorithms can be classified into high-quality but superlinear multilevel algorithms, and cheaper linear time approaches such as single-level partitioning and streaming partitioning.
We demonstrate both in theory and in practice that it is possible to achieve the best of both worlds at once.
Our linear time multilevel algorithm uses edge sparsification to constrain the size of subsequent coarser levels, which provably guarantees linear work while maintaining scalability to many cores.
We minimize quality loss by choosing appropriate thresholds for triggering the sparsification step and, if triggered, removing the edges with lowest weight.
As a result, our multilevel algorithm is faster than state-of-the-art single-level and streaming approaches while consistently computing better solutions
-- making multilevel the preferable choice even if extremely short running time is required.

\bibliography{references}

\newpage
\appendix

\section{Weighted Forest Fire}\label{app:wff}

\begin{algorithm2e}[h]
    \caption{
        Weighted Forest Fire: graph $G = (V, E)$, burn ratio $\nu$, probability $p$.
        The only difference to the original Forest Fire~\cite{ForestFire} algorithm is highlighted \textcolor{black!20!blue}{blue}.
    }
    \label{alg:wff}
    $\mathcal{S} \gets \textrm{new}~\FuncSty{Array}() \textrm{~of size $\lvert E \rvert$}$ \tcp*{Scores}
    $\FuncSty{b} \gets 0$ \tcp*{Number of burnt edges}
    \ParallelWhile{$\FuncSty{b} \le \nu \lvert E \rvert$}{ \label{line:loop}
        $Q \gets \textrm{new}~\FuncSty{FIFO}()$; $Q.\FuncSty{push}($random node from $V)$ \tcp*{BFS queue} \label{line:bfsinit}
        $T \gets \textrm{new}~\FuncSty{Set}()$ \tcp*{Visited nodes}
        \While{$Q \neq \emptyset$}{
            $u \gets Q.\FuncSty{pop}()$ \;
            \While{$N(u) \setminus T \neq \emptyset$}{
                Sample $v$ from $N(u) \setminus T$ with \textcolor{black!20!blue}{prob. $\omega(\{u, v\}) / \sum_{v \in N(u) \setminus T} \omega(\{u, v\})$} \label{line:weighted} \;
                $T.\FuncSty{insert}(v)$; $Q.\FuncSty{push}(v)$ \;
                $\mathcal{S}[\{u, v\}] \pluseqatomic 1$; $\FuncSty{b} \pluseqatomic 1$ \label{line:burn} \;
                \textbf{Break} with prob. $p$ \;
            }
        }
    }
    \Return{$\mathcal{S}$}
\end{algorithm2e}

\Cref{alg:wff} shows a modified version of Forest Fire that incorporates edge weights into the scoring process.
Following the original algorithm, it computes edge scores by simulating fires spreading through the graph via multiple traversals starting from random nodes (line~\ref{line:bfsinit}).
When visiting a node $u$, the number of neighbors $X$ to be visited is drawn from a geometric distribution parameterized by a tunable parameter $p$.
Subsequently, $X$ distinct unvisited neighbors of $u$ are sampled, incrementing the \emph{burn scores} of the corresponding edges and spreading the fire to the sampled nodes.
We incorporate edge weights by making this sampling weight-dependent: the probability of selecting neighbor $v$ of $u$ is proportional to the edge weight $\omega(\{u, v\})$ relative to the total edge weight between $u$ and its unvisited neighbors (line \ref{line:weighted}).
Note that this modification (marked blue in \Cref{alg:wff}) is the only difference to the normal Forest Fire algorithm.
Each edge traversal increments the burn score of the edge (line~\ref{line:burn}).
The algorithm stops scheduling additional fires once the cumulative burn score
$\FuncSty{b}$ 
exceeds $\nu \lvert E \rvert$
(line~\ref{line:loop})
for some tunable \emph{burn ratio} $\nu > 0$.

\section{Deep Multilevel Partitioning in the Worst Case}\label{app:deep-multilevel}

The default mode of \Partitioner{KaMinPar} uses a scheme known as \emph{deep multilevel partitioning} for coarsening and initial partitioning~\cite{deep-mgp}.
However, as discussed in the following, deep multilevel partitioning without additional modifications is unfit to achieve linear running time.
The problem is that expensive bipartitioning algorithms might be applied to relatively large graphs due to a combination of two effects.
First, the coarsening algorithm applies no size constraints or early stopping during the clustering (except for the block weight).
This can result in a massive size reduction at each coarsening step.
Second, the deep multilevel scheme applies bipartitioning steps over multiple levels if the current level $G_i$ is below a size threshold -- in this case, bipartitioning is applied to the next larger level $G_{i-1}$ (this is motivated by scalability; refer to the original publication for technical details~\cite{deep-mgp}).
While this is not a problem in itself, the fact that there is no bound on the size of $G_{i-1}$ means that bipartitioning might be applied to a graph with $\Th{n}$ nodes, possibly even to the input graph.\footnote{
    The original analysis uses simplified assumptions for coarsening~\cite{deep-mgp}, thereby sidestepping this problem. However, in practice we observed notable slowdowns due to this effect.
}
Since bipartitioning is crucial for solution quality, it needs to include the more expensive FM refinement algorithm in practice.
FM refinement has $\Omega(n \log n)$ running time, thereby adding a $\Omega(n \log n)$ term to the total running time in the worst case.

\section{Details on the Benchmark Sets}\label{app:benchmark-set}

\begin{table}[h]
    \caption{
        The benchmark graphs with their number of nodes $n$ and number of undirected edges $m$.
        Graphs included in the tuning subset are \textbf{bolded}.
    }
    \label{benchmark-table}

    \centering

    \renewcommand{\arraystretch}{0.80}

    \begin{tabular}{lrrlll}
        Graph & $n$ & $m$ & Class & Ref. \\
        \midrule
        \Instance{kmerV1r} & \numprint{214004392} & \numprint{232704832} & Kmer & \cite{mt-kahypar-ufm, SPM} \\
        \Instance{kmerA2a} & \numprint{170372459} & \numprint{179941739} & Kmer & \cite{mt-kahypar-ufm, SPM} \\
        \textbf{\Instance{kmerP1a}} & \numprint{138896082} & \numprint{148465346} & Kmer & \cite{mt-kahypar-ufm, SPM} \\
        \Instance{kmerU1a} & \numprint{64678340} & \numprint{66393629} & Kmer & \cite{mt-kahypar-ufm, SPM} \\
        \Instance{kmerV2a} & \numprint{53500237} & \numprint{57076126} & Kmer & \cite{mt-kahypar-ufm, SPM} \\
        \midrule
        \Instance{recompProteins1GB-9} & \numprint{14537567} & \numprint{74308567} & Compression & \cite{mt-kahypar-ufm, TextRecompression} \\
        \Instance{recompProteins1GB-7} & \numprint{2825742} & \numprint{43857382} & Compression & \cite{mt-kahypar-ufm, TextRecompression} \\
        \textbf{\Instance{recompDNA1GB-9}} & \numprint{3233125} & \numprint{25285086} & Compression & \cite{mt-kahypar-ufm, TextRecompression} \\
        \Instance{recompSources1GB-9} & \numprint{2792175} & \numprint{12188301} & Compression & \cite{mt-kahypar-ufm, TextRecompression} \\
        \Instance{recompSources1GB-7} & \numprint{898704} & \numprint{7272363} & Compression & \cite{mt-kahypar-ufm, TextRecompression} \\
        \midrule
        \Instance{comFriendster} & \numprint{65608366} & \numprint{1806067135} & Social & \cite{mt-kahypar-ufm, SNAP} \\
        \Instance{twitter2010} & \numprint{41652230} & \numprint{1202513046} & Social & \cite{mt-kahypar-ufm, networkrepository} \\
        \Instance{socSinaweibo} & \numprint{58655849} & \numprint{261321033} & Social & \cite{mt-kahypar-ufm, networkrepository} \\
        \Instance{comOrkut} & \numprint{3072627} & \numprint{117185083} & Social & \cite{mt-kahypar-ufm, SNAP} \\
        \Instance{comLJ} & \numprint{4036538} & \numprint{34681189} & Social & \cite{mt-kahypar-ufm, SNAP} \\
        \Instance{socFlickr} & \numprint{1715255} & \numprint{15555041} & Social & \cite{mt-kahypar-ufm, networkrepository} \\
        \Instance{wikiTalk} & \numprint{2394385} & \numprint{4659565} & Social & \cite{SNAP} \\
        \textbf{\Instance{comDBLP}} & \numprint{317080} & \numprint{1049866} & Social & \cite{SNAP} \\
        \Instance{coAuthorsDBLP} & \numprint{299067} & \numprint{977676} & Social & \cite{SNAP} \\
        \midrule
        \Instance{chungLuN27M30} & \numprint{134217728} & \numprint{600116012} & ChungLu & \cite{NetworKitNew} \\
        \Instance{chungLuN22M30} & \numprint{4194304} & \numprint{19842651} & ChungLu & \cite{NetworKitNew} \\
        \textbf{\Instance{chungLuN22M25}} & \numprint{4194304} & \numprint{12220794} & ChungLu & \cite{NetworKitNew} \\
        \Instance{chungLuN17M25} & \numprint{131072} & \numprint{511265} & ChungLu & \cite{NetworKitNew} \\
        \midrule
        \textbf{\Instance{erN24M29}} & \numprint{16777216} & \numprint{536870912} & Erd\H{o}s-R\'enyi & \cite{KaGen} \\
        \Instance{erN23M28} & \numprint{8388608} & \numprint{268435456} & Erd\H{o}s-R\'enyi & \cite{KaGen} \\
        \Instance{erN21M25} & \numprint{2097152} & \numprint{33554432} & Erd\H{o}s-R\'enyi & \cite{KaGen} \\
        \Instance{erN21M24} & \numprint{2097152} & \numprint{16777216} & Erd\H{o}s-R\'enyi & \cite{KaGen} \\ 
        \midrule
        \textbf{\Instance{plantedN22K5M29}} & \numprint{4194304} & \numprint{536871319} & Planted & \cite{NetworKitNew} \\
        \Instance{plantedN18K5M26} & \numprint{262144} & \numprint{67114407} & Planted & \cite{NetworKitNew} \\
        \Instance{plantedN20K5M26} & \numprint{1048576} & \numprint{67105883} & Planted & \cite{NetworKitNew} \\
        \Instance{plantedN22K7M25} & \numprint{4194304} & \numprint{33560609} & Planted & \cite{NetworKitNew} \\
        \midrule
        \Instance{rmatN24M29-0.5-0.3-0.1} & \numprint{16777216} & \numprint{533964993} & R-MAT & \cite{KaGen} \\
        \Instance{rmatN25M28} & \numprint{27089643} & \numprint{268415704} & R-MAT & \cite{mt-kahypar-ufm, KaGen} \\
        \textbf{\Instance{rmatN24M28-0.5-0.3-0.1}} & \numprint{16777216} & \numprint{267658930} & R-MAT & \cite{KaGen} \\
        \Instance{rmatN23M27-0.5-0.3-0.1} & \numprint{8388608} & \numprint{133683000} & R-MAT & \cite{KaGen} \\
        \midrule
        \Instance{baN22} & \numprint{4194304} & \numprint{8388529} & Other & \cite{KaGen} \\
        \Instance{kktPower} & \numprint{2063494} & \numprint{6482320} & Other & \cite{SPM} \\
        \Instance{debrG18} & \numprint{1048576} & \numprint{2097149} & Other & \cite{SPM} \\
        \textbf{\Instance{debrG17}} & \numprint{524288} & \numprint{1048573} & Other & \cite{SPM} \\
        \bottomrule
    \end{tabular}
\end{table}

\end{document}